\documentclass[11pt, a4paper, oneside]{article}

\usepackage[margin = 1.3in]{geometry}

\def\citep{\cite}

\usepackage{times}
\usepackage{amsfonts,psfrag,epsfig}
\usepackage{amsfonts,epsfig}
\usepackage{color}
\usepackage{amsmath,amssymb,hyperref}
\usepackage[amsmath,thmmarks]{ntheorem}
\usepackage{tabularx}

\usepackage{tikz}
\usetikzlibrary{intersections}
\usetikzlibrary{shapes,arrows}
\usepackage{subfigure}

\hypersetup{colorlinks=true}
\newtheorem{theorem}{Theorem}
\newtheorem{lemma}[theorem]{Lemma}
\newtheorem{corollary}[theorem]{Corollary}

\newtheorem{claim}[theorem]{Claim}

\newtheorem{definition}{Definition}

{ \theoremstyle{nonumberplain}\theoremsymbol{\ensuremath{\Box}}
\theorembodyfont{\normalfont}
\newtheorem{proof}{Proof.}

\theoremsymbol{\ensuremath{\Box}}
\theoremheaderfont{\normalfont\itshape}
\theorembodyfont{\normalfont} \theoremstyle{empty}

}

\newcommand{\union}{\cup}

\newcommand{\beq}{\begin{eqnarray}}
\newcommand{\eeq}{\end{eqnarray}}
\newcommand{\beqn}{\begin{equation}}
\newcommand{\eeqn}{\end{equation}}

\tikzstyle{decision} = [diamond,
                        draw,
                        fill=blue!20,
                        text width=1.5em,
                        text badly centered,
                        node distance=2cm,
                        inner sep=0pt]
\tikzstyle{block} = [rectangle,
                     draw,
                     fill=blue!20,
                     text centered,
                     rounded corners,
                     minimum height=2em]
\tikzstyle{line} = [draw,
                    -latex']
\tikzstyle{cloud} = [draw,
                     ellipse,
                     fill=red!20,
                     node distance=2cm,
                     minimum height=2em]

\allowdisplaybreaks

\begin{document} 

\title{Minimum Weight Perfect Matching\\via Blossom Belief Propagation}

\author{Sungsoo Ahn, $\quad$ Sejun Park, $\quad$  Michael Chertkov\thanks{M.\ Chertkov is with the Theoretical Division at Los Alamos National Laboratory, USA.
Author's e-mail: \texttt{chertkov@lanl.gov}}, $\quad$  Jinwoo Shin\thanks{S.\ Ahn, S.\ Park and J.\ Shin are with the Department of Electrical Engineering at Korea Advanced Institute of Science Technology, Republic of Korea.
Authors' e-mails: \texttt{ssahn0215@kaist.ac.kr, sejun.park@kaist.ac.kr, jinwoos@kaist.ac.kr}
}}


\maketitle




\begin{abstract}
Max-product Belief Propagation (BP) is a popular message-passing algorithm for 
computing a Maximum-A-Posteriori (MAP) assignment over a distribution represented 
by a Graphical Model (GM). It has been shown that BP can solve a number of combinatorial 
optimization problems including minimum weight matching, shortest path, network flow and vertex cover
under the following common assumption: the respective Linear Programming (LP) relaxation is tight, i.e., 
no integrality gap is present.  However, when LP shows an integrality gap,
no model has been known 
which can be solved systematically via sequential applications of BP.
In this paper, we develop the first such algorithm, coined Blossom-BP, for
solving the minimum weight matching problem over arbitrary graphs. 
Each step of the sequential algorithm requires applying BP 
over a modified graph constructed by contractions and expansions of blossoms, i.e., odd sets of vertices. 
Our scheme guarantees termination in $O(n^2)$ of BP runs, where $n$ is the 
number of vertices in the original graph. In essence, the Blossom-BP offers a distributed version of the celebrated 
Edmonds' Blossom algorithm by jumping at once over many sub-steps with a single BP.
Moreover, our result provides an interpretation of the Edmonds' algorithm as a sequence of LPs.
\end{abstract}

\vspace{-0.1in}
\section{Introduction}
Graphical Models (GMs) provide a useful representation for reasoning in a number of scientific disciplines \cite{05YFW,08RU,09MM,08WJ}. Such models use a graph structure to encode the joint probability distribution, where vertices correspond to random variables and edges specify conditional dependencies. An important inference task in many applications involving GMs is to find the most-likely assignment to the variables in a GM, i.e., Maximum-A-Posteriori (MAP). Belief Propagation (BP) is a popular algorithm for approximately solving the MAP inference problem and it is an iterative, message passing one 
that is exact on tree structured GMs. BP often shows remarkably strong heuristic performance beyond trees, i.e., over loopy GMs. Furthermore, BP is 
of a particular relevance to large-scale problems due to its potential for parallelization \cite{09GLG} and
its ease of programming within the modern programming models for parallel computing, e.g., GraphLab \cite{10LGKBGH}, GraphChi \cite{12KBC} and OpenMP \cite{20CMDK}.

The convergence and correctness of BP was recently established for a certain class of loopy GM formulations of several classical combinatorial optimization problems, including matching \cite{08BSS,11SMW,07HJ}, perfect matching \cite{11BBCZ}, shortest path \cite{08RT}, independent set \cite{07SSW}, network flow \cite{10GSY} and vertex cover \cite{PS145}. The important common feature of these models is that BP converges to a correct assignment when the Linear Programming (LP) relaxation of the combinatorial
optimization is tight, i.e., when it shows no integrality gap. The LP tightness is an inevitable condition to guarantee the performance of BP and no combinatorial optimization instance has been known where BP would be used 
to solve problems without the LP tightness. On the other hand, in the LP literature, it has been extensively studied how to enforce the LP tightness via solving multiple intermediate LPs that are systematically designed, e.g., via the cutting-plane method \cite{12CLS}. Motivated by these studies, we pose a similar question for BP, ``how to enforce correctness of BP, possibly by solving multiple intermediate BPs''.
In this paper, 
we show how to resolve this question for the minimum weight (or cost) perfect matching problem over arbitrary graphs.


\paragraph{Contribution.}
We develop an algorithm, coined Blossom-BP, for
solving the minimum weight matching problem over an arbitrary graph.
Our algorithm solves multiple intermediate BPs until the final BP outputs the solution.
The algorithm 
is sequential, where each step includes running BP over a `contracted' graph 
derived from the original graph by contractions and infrequent expansions of blossoms, i.e., odd sets of vertices. 
To build such a scheme, we first design an algorithm, coined Blossom-LP,
solving multiple intermediate LPs.
Second, we show that each LP is solvable by BP using 
the recent framework \cite{PS145} that establishes a generic connection between
BP and LP. For the first part, cutting-plane methods solving multiple intermediate LPs for the minimum weight matching problem
have been discussed by several authors over the past decades \cite{78T, 82PR, 85GH, 86LP, 07FL} and a provably polynomial-time scheme was
recently suggested \cite{12CLS}. However, LPs in \cite{12CLS} were quite complex to solve by BP. To address 
the issue, we design much simpler intermediate LPs that allow utilizing the framework of \cite{PS145}.

We prove that Blossom-BP and Blossom-LP guarantee to terminate in $O(n^2)$ of BP and LP runs, respectively,
where $n$ is the number of vertices in the graph.  
To establish the polynomial complexity, we show that
intermediate outputs of Blossom-BP and Blossom-LP
are equivalent to those of
a variation of
the Blossom-V algorithm \cite{K09}
which is the latest implementation of the Blossom algorithm due to Kolmogorov.
The main difference is that Blossom-V updates parameters by maintaining disjoint tree graphs,
while Blossom-BP and Blossom-LP implicitly achieve this by maintaining disjoint cycles, claws and tree graphs.
Notice, however, that these combinatorial structures are auxiliary, as required for proofs,  and they 
do not appear explicitly in the algorithm descriptions.
Therefore, they are 
much easier to implement than Blossom-V that maintains complex data structures, e.g., priority queues.
To the best of our knowledge, Blossom-BP and Blossom-LP are the simplest possible algorithms available for solving the problem in polynomial time.
Our proof implies that
in essence, Blossom-BP offers a distributed version of the Edmonds' Blossom algorithm 
\cite{65Edm}
jumping at once over many sub-steps of Blossom-V with a single BP.

The subject of solving convex optimizations (other than LP) via BP was discussed in the literature \cite{MJW06, WYM07, MR08}.
However, we are not aware of any similar attempts to solve
Integer Programming, via sequential application of BP.
We believe that the approach developed in this paper is of a
broader interest, as it promises to advance the challenge of designing 
BP-based MAP solvers for a broader class of GMs.
Furthermore, Blossom-LP stands alone as providing 
an interpretation for the Edmonds' algorithm 
in terms of a sequence of tractable LPs.
The Edmonds' original LP formulation contains exponentially many constraints,
 thus naturally suggesting  to seek for 
a sequence of LPs, each with a subset of constraints, gradually reducing the integrality gap to zero in a polynomial number of steps.
However, it remained illusive for decades: 
even when the bipartite LP relaxation of the problem has an integral optimal solution,
the standard Edmonds' algorithm keeps contracting and expanding a sequence of blossoms.
As we mentioned earlier, we resolve the challenge by  showing that Blossom-LP is (implicitly) equivalent to a variant of 
the Edmonds' algorithm with 
three major modifications: (a) parameter-update via maintaining cycles, claws and trees,
(b) addition of small random corrections to weights, and 
(c) initialization using the bipartite LP relaxation.

\paragraph{Organization.} In Section \ref{sec:pre}, we provide backgrounds on the minimum weight perfect matching problem
and the BP algorithm. Section \ref{sec:main} describes our main result -- Blossom-LP and Blossom-BP algorithms, where
the proof is given in Section \ref{sec:proof}.

\vspace{-0.05in}
\section{Preliminaries}\label{sec:pre}
\vspace{-0.05in}
\subsection{Minimum weight perfect matching}
Given an (undirected) graph $G=(V,E)$,  
a matching of $G$ is a set of vertex-disjoint edges, where
a perfect matching additionally requires to cover every vertices of $G$.
Given integer edge weights (or costs) $w=[w_e]\in\mathbb{Z}^{|E|}$, 
the minimum weight (or cost) perfect matching problem consists in computing a perfect matching
which minimizes the summation of its associated edge weights. The problem is formulated as the following IP (Integer Programming):
\begin{equation}\label{ip:matching}
\begin{split}
\mbox{minimize}\qquad&w\cdot x\qquad
\mbox{subject to}\qquad\sum_{e\in\delta(v)}x_e=1,\quad\forall v\in V,\qquad x=[x_e]\in\{0,1\}^{|E|}
\end{split}
\end{equation}
Without loss of generality, one can assume that weights are strictly positive.\footnote{If some edges have negative weights,
one can add the same positive constant to all edge weights, and this does not alter the solution of IP \eqref{ip:matching}.}
Furthermore, we assume that IP \eqref{ip:matching} is feasible, i.e., there exists at least one perfect matching in $G$.
One can naturally relax the above integer constraints to $x=[x_e]\in [0,1]^{|E|}$
to obtain an LP (Linear Programming), which is called the bipartite relaxation.
The
integrality of the bipartite LP relaxation is not guaranteed, however it can be enforced by adding the so-called blossom inequalities \cite{K09}:
\begin{equation}\label{lp:edmonds}
\begin{split}
\mbox{minimize}\qquad&w\cdot x\\
\mbox{subject to}\qquad&\sum_{e\in\delta(v)}x_e=1,\quad\forall v\in V,\qquad
\sum_{e\in\delta(S)}x_e\ge 1,\quad\forall S\in\mathcal{L},\qquad x=[x_e]\in[0,1]^{|E|},
\end{split}
\end{equation}
where $\mathcal{L}\subset 2^V$ is a collection of odd cycles in $G$, called blossoms,
and $\delta(S)$ is a set of edges between $S$ and $V\setminus S$.
It is known that if $\mathcal{L}$ 
is the
collection of all the odd cycles in $G$, then LP \eqref{lp:edmonds} always has
an integral solution.
However, notice that the number of odd cycles is exponential in $|V|$, thus solving 
LP \eqref{lp:edmonds} is computationally intractable.
To overcome 
this complication we are looking for a tractable subset of $\mathcal L$ of a polynomial size which guarantees the integrality. 
Our algorithm, searching for such a tractable subset of $\mathcal L$ is iterative: at each iteration it adds or subtracts a blossom.

\subsection{Background on max-product Belief Propagation}

The max-product Belief Propagation (BP) algorithm is a popular heuristic for approximating the MAP assignment in a GM.
BP is implemented iteratively; at each iteration $t$, it maintains four 
messages 
$$\{m^{t}_{\alpha\rightarrow i}(c), 
m^{t}_{i\rightarrow\alpha}(c):
c\in\{0,1\} \}$$
 between
every variable $z_i$ and every associated $\alpha\in F_i$, where
 $F_i:= \{\alpha\in F: i \in \alpha\}$; that is, $F_i$ is
a subset of $F$ such that all $\alpha$ in $F_i$ include the $i^{th}$
position of $z$ for any given $z$.
The messages are updated as follows:
\begin{align}
&\quad m^{t+1}_{\alpha\rightarrow i}(c) ~=~ \max_{z_{\alpha}:z_i=c}
                                   \psi_\alpha (z_{\alpha})
\prod_{j\in \alpha\setminus i} m_{j\rightarrow \alpha}^t (z_j)\label{eq:msg:alpha_to_i}\\
&\quad m^{t+1}_{i\rightarrow\alpha}(c) ~=~
\psi_i(c)\prod_{\alpha^{\prime}\in F_i\setminus \alpha} m_{\alpha^{\prime}
\rightarrow i}^t (c)\label{eq:msg:i_to_alpha}.
\end{align}
%
where each $z_i$ only sends messages to $F_i$; that is,
$z_i$ sends messages to $\alpha_j$ only if $\alpha_j$ selects/includes $i$.
%
%
%
The outer-term in the message computation \eqref{eq:msg:alpha_to_i} 
is maximized over all possible $z_\alpha\in\{0,1\}^{|\alpha|}$ with $z_i=c$.
%
%
The inner-term is a product that only depends on the variables $z_j$ (excluding $z_i$) that
are connected to $\alpha$. 
%
%
%
%
The message-update \eqref{eq:msg:i_to_alpha} from variable $z_i$ to factor $\psi_\alpha$ is a product containing
all messages received by $\psi_\alpha$ in the previous iteration,
except for the message sent by $z_i$ itself.
%
%

Given a set of messages $\{m_{i\to\alpha}(c)$, $m_{\alpha\to
i}(c):c\in\{0,1\}\}$, the so-called BP marginal beliefs are computed as follows:
\begin{eqnarray}\label{eq:bpdecision}
b_i[z_i]&=&
{\psi_i (z_i)}
\prod_{\alpha\in F_i} m_{\alpha\to i}(z_i).\label{eq:marginalbelief}
%
%
%
%
\end{eqnarray}
This BP algorithm outputs $z^{BP}=[z_i^{BP}]$ where
$$
z_i^{BP}=\begin{cases}
1&\mbox{if}~ b_i[1]>b_i[0]\\
?&\mbox{if}~b_i[1]=b_i[0]\\
0&\mbox{if}~ b_i[1]<b_i[0]
\end{cases}.
$$
It is known that $z^{BP}$ converges to a MAP assignment after a sufficient number of iterations,
if the factor graph 
is a tree and the MAP assignment is unique. However, if the graph contains loops, the BP
algorithm is not guaranteed to converge to a MAP assignment in general.

\vspace{-0.05in}
\subsection{Belief propagation for linear programming}

A joint distribution of $n$ (binary) random variables $Z=[Z_i]\in \{0,1\}^n$ is called a Graphical Model (GM) if it factorizes as follows: for $z=[z_i]\in \Omega^n$,
\begin{equation*}
	\Pr[Z=z]~\propto~\prod_{i\in\{1,\dots,n\}}\psi_i(z_i)\prod_{\alpha\in F} \psi_{\alpha} (z_\alpha),\label{eq:generic_gm}
\end{equation*}
where $\{\psi_i,\psi_{\alpha}\}$ are (given) non-negative functions, the so-called factors; 
$F$ is a collection of subsets 
$$F=\{\alpha_1,\alpha_2,...,\alpha_k\}\subset 2^{\{1,2,\dots, n\}}$$
(each $\alpha_j$ is a subset of $\{1,2,\dots, n\}$ with $|\alpha_j|\ge 2$); $z_\alpha$ is the projection of $z$ onto
dimensions included in $\alpha$.\footnote{For example, if $z=[0,1,0]$ and $\alpha=\{1,3\}$, then $z_\alpha=[0,0]$.} 
In particular, $\psi_i$ is called a variable factor.
%
%
%
Assignment ${z}^*$ is called a maximum-a-posteriori (MAP) solution if
${z}^*=\arg\max_{{z}\in\{0,1\}^n} \Pr[{z}].$
Computing a MAP solution is typically
computationally intractable (i.e., NP-hard) unless the induced bipartite graph 
of factors $F$ and variables $z$, so-called factor graph, has a bounded treewidth \citep{CSH08}.
%
%
%
%
The max-product Belief Propagation (BP) algorithm is a popular simple heuristic for approximating the MAP solution in a GM,
where it iterates messages over a factor graph. 
BP computes a MAP solution exactly after a sufficient number of iterations,
if the factor graph 
is a tree and the MAP solution is unique. However, if the graph contains loops, BP
is not guaranteed to converge to a MAP solution in general.
Due to the space limitation, we provide detailed backgrounds
on BP in the supplemental material.

Consider the following GM:  for $x=[x_i]\in \{0,1\}^n$ and $w=[w_i]\in \mathbb R^n$,
\begin{equation}
	\Pr[X=x]~\propto~\prod_{i} e^{-w_i x_i}\prod_{\alpha\in F} \psi_{\alpha} (x_\alpha),\label{eq:gm1}
\end{equation}
where $F$ is the set of non-variable factors and
the factor function $\psi_\alpha$ for $\alpha\in F$ is defined as
\begin{align*}
&\psi_{\alpha}(x_{\alpha}) = 
\begin{cases}
1&\mbox{if}~ A_{\alpha} x_{\alpha}\ge b_{\alpha},~ C_{\alpha}x_{\alpha}=d_{\alpha}\\
0&\mbox{otherwise}
\end{cases},
\end{align*}
for some
matrices $A_{\alpha}, C_{\alpha}$ and vectors $b_{\alpha}, d_{\alpha}$.
Now we consider the Linear Program (LP) corresponding to this GM:
\begin{equation}\label{eq:lp1}
\begin{split}
	&\mbox{minimize}\qquad~ w\cdot x\\
	&\mbox{subject to}\qquad 
\psi_\alpha(x_\alpha)=1,\quad \forall \alpha\in F,\qquad
x=[x_i]\in [0,1]^n.
\end{split}
\end{equation}
One observes that the MAP solution for GM \eqref{eq:gm1} corresponds to the (optimal) solution of LP \eqref{eq:lp1}
if the LP has an integral solution $x^*\in \{0,1\}^n$. 
Furthermore, the following sufficient conditions
relating max-product BP to LP are known \cite{PS145}:

\begin{theorem}\label{thm:bplp}
      The max-product BP applied to GM \eqref{eq:gm1} 
converges to the solution of LP \eqref{eq:lp1} if the following conditions hold:
\begin{itemize}
\item[C1.] LP \eqref{eq:lp1} has a unique integral solution $x^*\in\{0,1\}^n$, i.e., it is tight.
\item[C2.] For every $i\in \{1,2,\dots, n\}$, the number of factors associated with $x_i$ is at most two, i.e., $|F_i|\leq 2.$
\item[C3.] For every factor $\psi_\alpha$, every $x_\alpha\in\{0,1\}^{|\alpha|}$ with $\psi_\alpha(x_\alpha)=1$, and 
	every $i\in\alpha$ with $x_i\neq x^*_i$, 
	there exists 
$\gamma\subset \alpha$ 
such that
$$|\{j \in\{i\}\union \gamma:|F_j|=2\}|\le 2$$
$$
\psi_\alpha(x^\prime_\alpha)=1,\qquad
\mbox{
where 
$x^\prime_k = \begin{cases}
		x_k~&\mbox{if}~k\notin \{i\}\union \gamma\\
		x^*_k~&\mbox{otherwise}
\end{cases}$.}$$
$$
\psi_\alpha(x^{\prime\prime}_\alpha)=1,\qquad
\mbox{
where 
$x^{\prime\prime}_k = \begin{cases}
		x_k~&\mbox{if}~k\in\{i\}\union \gamma\\
		x^*_k~&\mbox{otherwise}
\end{cases}$.}$$
    
\end{itemize}       
\end{theorem}

\vspace{-0.1in}
\section{Main result: Blossom Belief Propagation}\label{sec:main}
In this section, we introduce our main result -- an iterative algorithm, coined Blossom-BP, for solving 
the minimum weight perfect matching problem over an arbitrary graph, where
the algorithm uses the max-product BP as a subroutine. We first describe the algorithm
using LP instead of BP in Section \ref{sec:blp}, where we call it Blossom-LP. 
Its BP implementation is explained in Section \ref{sec:bbp}.

\vspace{-0.05in}
\subsection{Blossom-LP algorithm}\label{sec:blp}
Let us modify the edge weights:
$w_{e} \leftarrow w_{e} + n_{e},$
where $n_{e}$ is an i.i.d.\ random number chosen in the interval $\left[0, \frac{1}{|V|}\right]$.
Note that 
the solution of the minimum weight perfect matching problem \eqref{ip:matching} 
remains the same after this modification {since sum of the overall noise is smaller than 1}.
The Blossom-LP algorithm updates the following parameters iteratively.
\begin{itemize}
\item[$\circ$] $\mathcal{L}\subset 2^V$: a {\em laminar} collection of odd cycles in $G$.
\item[$\circ$] $y_v,y_S$: $v\in V$ and $S\in \mathcal L$.
\end{itemize}
In the above, $\mathcal L$ is called laminar if 
{for every $S,T\in \mathcal L$, $S\cap T =\emptyset$, $S\subset T$ or $T\subset S$.}
We call $S\in \mathcal L$ an {\em outer} blossom if there exists no $T\in\mathcal L$
such that $S\subset T$.
Initially, $\mathcal L=\emptyset$ and $y_{v} = 0$ for all $v\in V$. 
The algorithm iterates between Step {\bf A} and Step {\bf B} and
terminates at Step {\bf C}.

\vspace{0.1in}
\noindent {\bf Blossom-LP algorithm}
\hrule
\paragraph {A. Solving LP on a contracted graph.}
First construct an auxiliary (contracted) graph $G^{\dagger} = (V^{\dagger}, E^{\dagger})$ 
by contracting every outer blossom in $\mathcal L$ to a single vertex,
where the weights $w^{\dagger}=[w^{\dagger}_e:e\in E^{\dagger}]$ are defined as
$$w_{e}^{\dagger} =  w_e-
\sum_{v\in V: v\not\in V^{\dagger},e\in \delta (v)}y_v
-\sum_{S\in \mathcal L: v(S)\not\in V^{\dagger},e\in \delta (S)}y_S, 
\qquad \forall~e\in E^{\dagger}.$$
We let $v(S)$ denote the blossom vertex in $G^\dagger$ coined as the contracted graph and
solve the following LP:
\begin{equation}\label{lp:modifiededmonds}
\begin{split}
\mbox{minimize}\qquad&w^{\dagger}\cdot x\\
\mbox{subject to}\qquad&
\sum_{e\in\delta(v)}x_e=1,\quad\forall~ v\in V^{\dagger},~\text{$v$ is a non-blossom vertex}\\
&\sum_{e\in\delta(v)}x_e\ge 1,\quad\forall~ v\in V^{\dagger},~\text{$v$ is a blossom vertex}\\
&x=[x_e]\in[0,1]^{|E^{\dagger}|}.
\end{split}
\end{equation}

\paragraph {B. Updating parameters.}
After we obtain a solution $x=[x_e:e\in E^{\dagger}]$ of LP \eqref{lp:modifiededmonds}, 
the parameters are updated as follows:
\begin{itemize}
\item[(a)] If $x$ is integral, i.e., $x\in \{0,1\}^{|E^{\dagger}|}$ and $\sum_{e\in\delta(v)}x_{e} = 1$ for all $v\in V^{\dagger}$, then proceed to the termination step {\bf C}. 
\item[(b)]
Else if there exists a blossom $S$ such that $\sum_{e\in\delta(v(S))}x_{e} > 1$, then we choose one of such 
blossoms and
update 
$$\mathcal L \leftarrow \mathcal L \backslash \{S\}\qquad\mbox{and}\qquad
y_v \leftarrow 0,\quad \forall ~v\in S.$$
Call this step `blossom $S$ expansion'.
\item[(c)] Else if there exists an odd cycle $C$ in $G^{\dagger}$ such that
$x_e=1/2$ for every edge $e$ in it, we choose one of them and update 
$$\mathcal L\leftarrow\mathcal{L}\cup \{V(C)\}
\qquad\mbox{and}\qquad
y_{v}\leftarrow {1\over 2}\sum_{e\in E(C)}(-1)^{d(e,v)}w^{\dagger}_{e}, \quad \forall v\in V(C),$$
where $V(C),E(C)$ are the set of vertices and edges of $C$, respectively,
and $d(v,e)$ is the graph distance from vertex $v$ to edge $e$ in the odd cycle $C$.
The algorithm also remembers the odd cycle $C=C(S)$ corresponding to every blossom $S\in \mathcal L$.
\end{itemize}
If (b) or (c) occur, go to Step {\bf A}. 

\paragraph {C. Termination.} 
The algorithm iteratively expands blossoms in $\mathcal L$ to obtain the minimum weighted perfect matching $M^{*}$ as follows:
\begin{itemize}
\item[(i)] Let $M^{*}$ be the set of edges in the original $G$ such that its corresponding edge $e$ in the contracted graph $G^\dagger$ has $x_{e} = 1$,
where $x=[x_e]$ is the (last) solution of LP \eqref{lp:modifiededmonds}.
\item[(ii)] If $\mathcal L = \emptyset$, output $M^{*}$. 
\item[(iii)] Otherwise, choose an outer blossom $S\in\mathcal L$, 
then update $G^{\dagger}$ by expanding $S$, i.e. $\mathcal L \leftarrow \mathcal L\backslash\{S\}$.
\item[(iv)] Let $v$ be the vertex in $S$ covered by $M^{*}$
and $M_{S}$ be a matching covering $S\backslash\{v\}$ using the edges of odd cycle $C(S)$. 
\item[(v)] Update $M^{*} \leftarrow M^{*}\cup M_{S}$ and go to Step (ii).
\end{itemize}
\vspace{0.1in}
\hrule
\vspace{0.1in}

\begin{figure}[h]
\centering
\subfigure[Initial graph]{
	\includegraphics[scale=0.2]{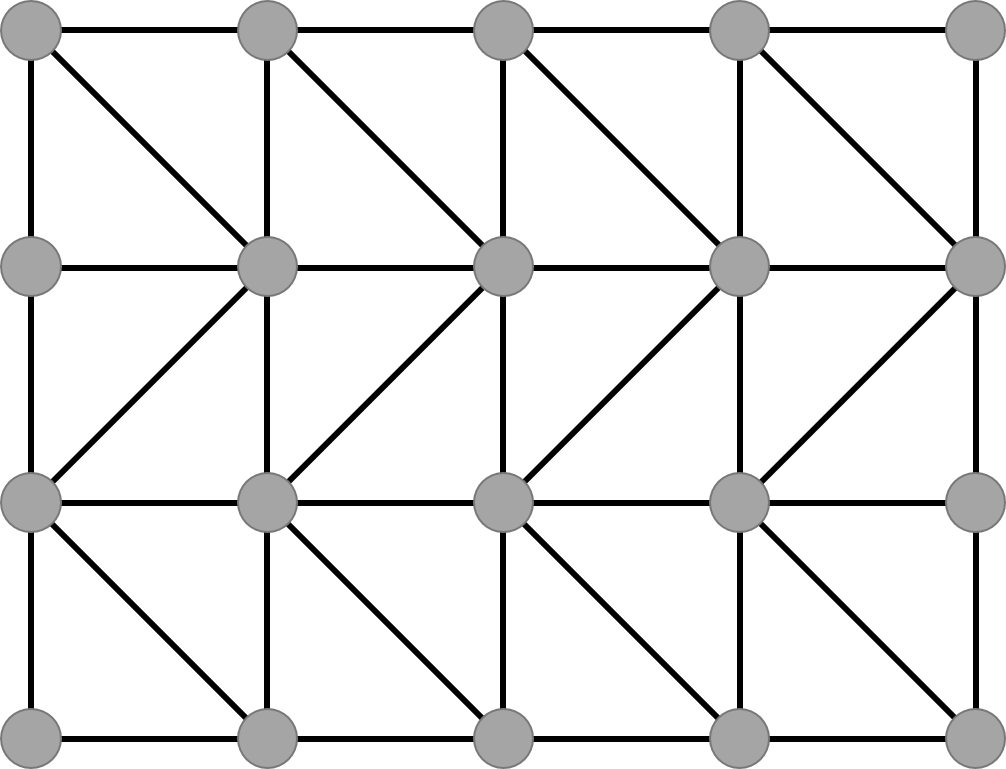}
	\label{fig:alg_step1}
}
\hfill
\subfigure[Solution of LP \eqref{lp:modifiededmonds} in the 1st iteration]{
	\includegraphics[scale=0.2]{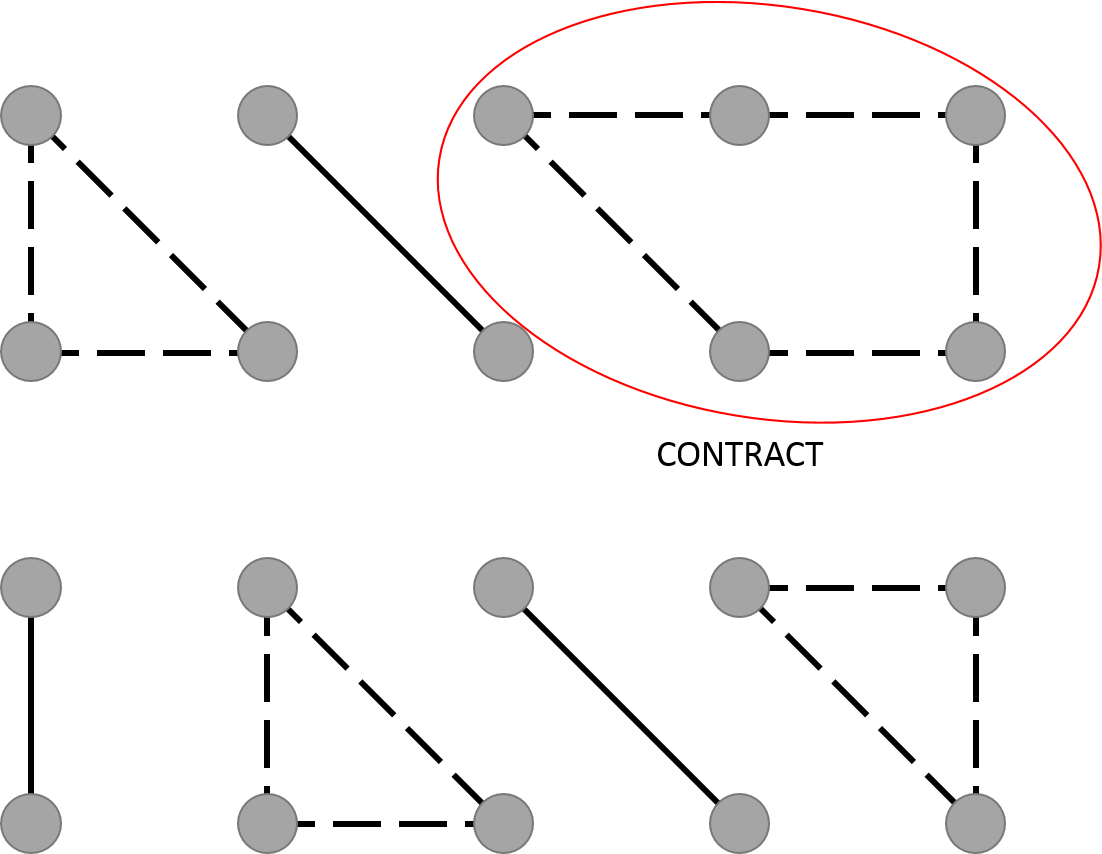}
	\label{fig:alg_step2}
}
\hfill
\subfigure[Solution of LP \eqref{lp:modifiededmonds} in the 2nd iteration]{
	\includegraphics[scale=0.2]{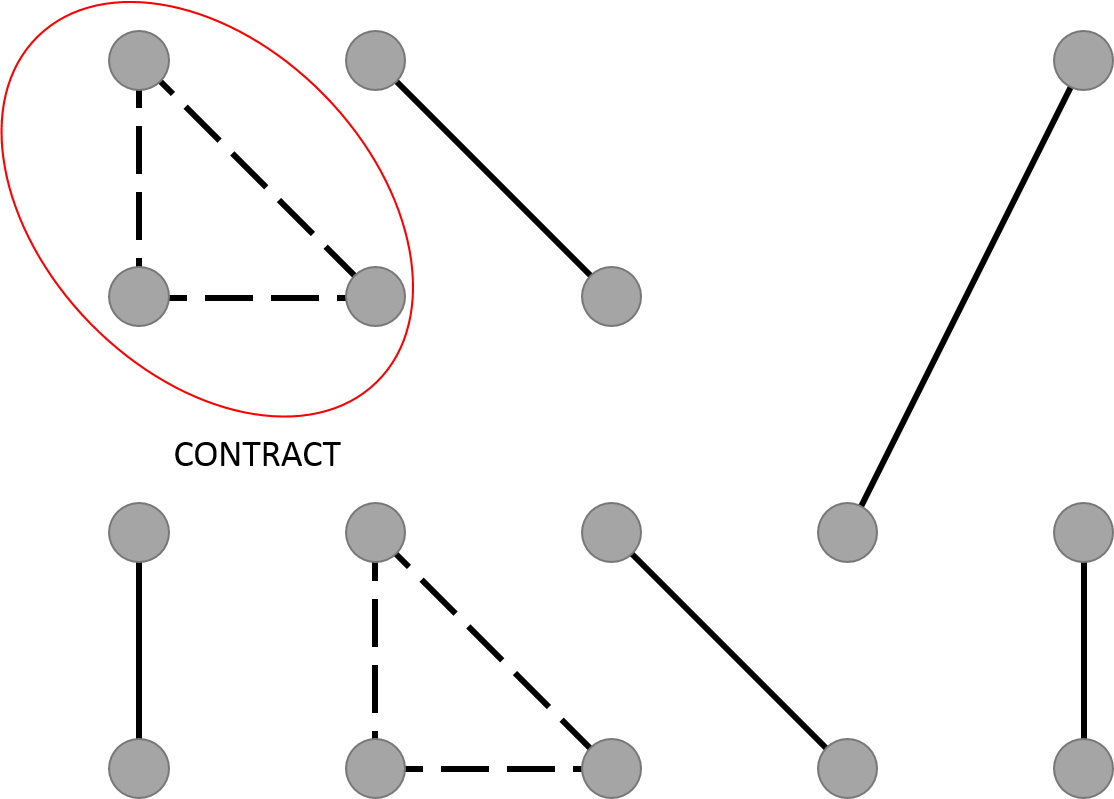}
	\label{fig:alg_step3}
}
\vspace*{\fill}
\subfigure[Solution of LP \eqref{lp:modifiededmonds} in the 3rd iteration]{
	\includegraphics[scale=0.2]{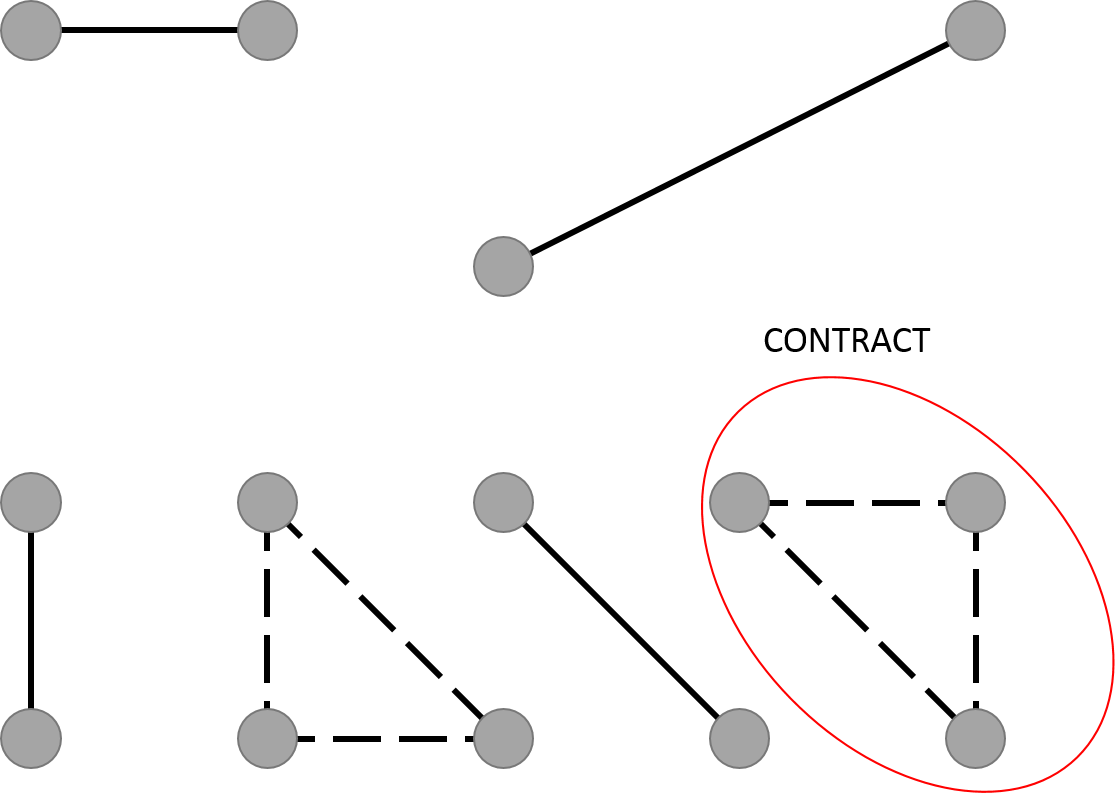}
	\label{fig:alg_step4}
}
\hfill
\subfigure[Solution of LP \eqref{lp:modifiededmonds} in the 4th iteration]{
	\includegraphics[scale=0.2]{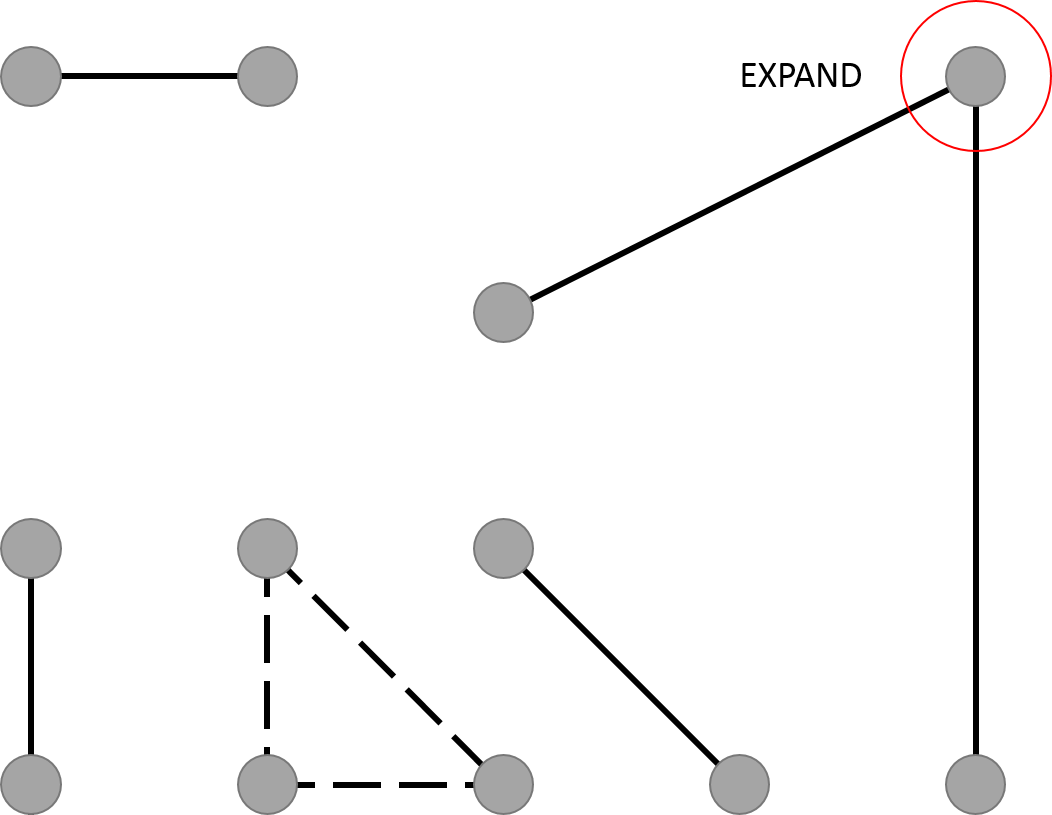}
	\label{fig:alg_step5}
}
\hfill
\subfigure[Solution of LP \eqref{lp:modifiededmonds} in the 5th iteration]{
	\includegraphics[scale=0.2]{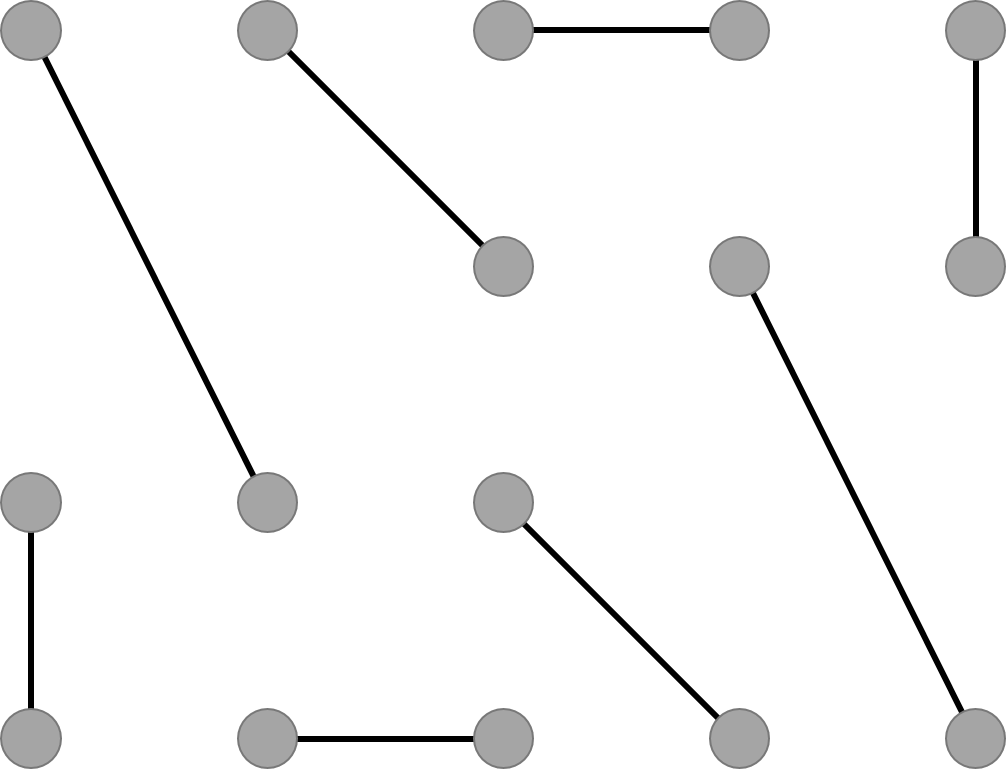}
	\label{fig:alg_step6}
}
\hfill
\subfigure[Output matching]{
	\includegraphics[scale=0.2]{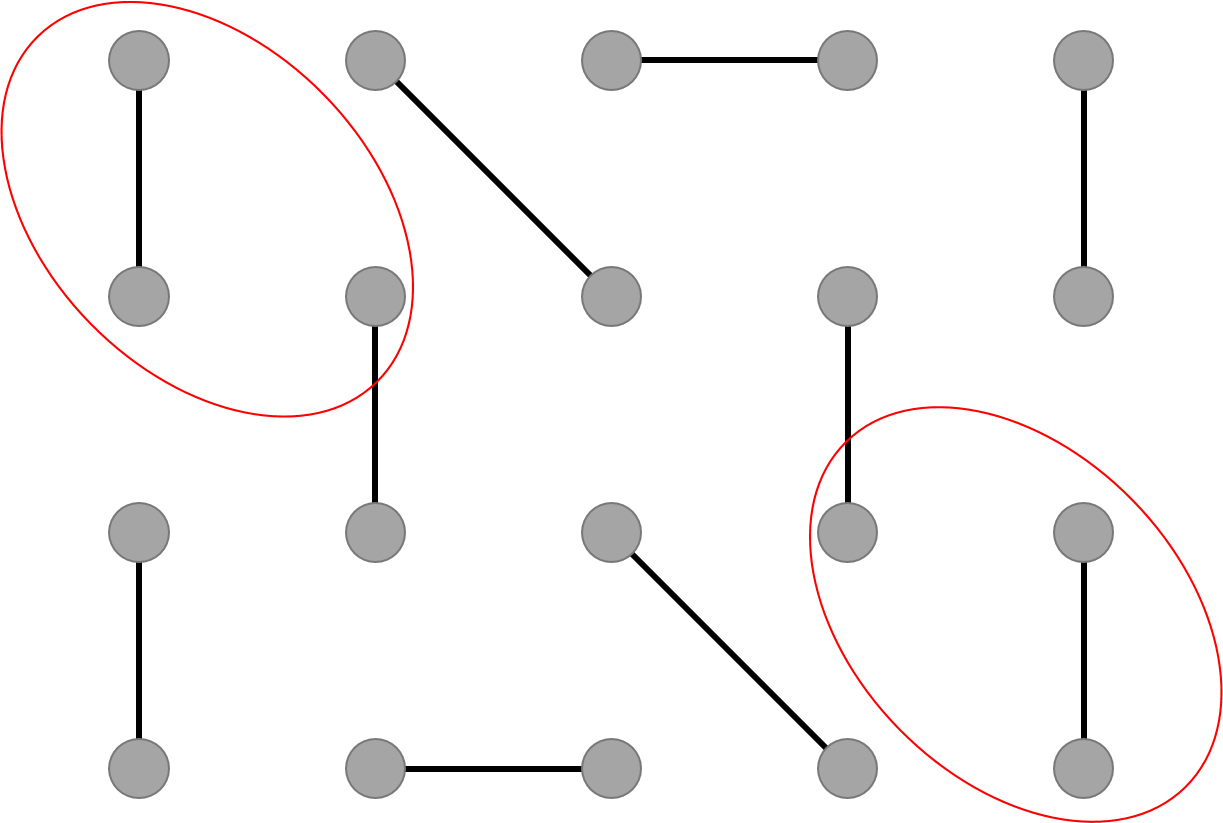}
	\label{fig:alg_step7}
}
\caption{Example of {evolution of Blossoms under Blossom-LP}, where
solid and dashed lines correspond to $1$ and $\frac{1}{2}$ solutions of LP \eqref{lp:modifiededmonds}, respectively.}
\end{figure}
We provide the 
following running time guarantee for this algorithm, which is proven in Section \ref{sec:proof}.
\begin{theorem}\label{thm:main}
Blossom-LP outputs the minimum weight perfect matching in 
$O(|V|^2)$ iterations.
\end{theorem}

\vspace{-0.05in}
\subsection{Blossom-BP algorithm}\label{sec:bbp}
In this section, we show that the algorithm can be implemented using BP. 
The result is derived in two steps,  where the first one consists in the following theorem.
\begin{theorem}\label{thm:halfintegral}
LP \eqref{lp:modifiededmonds}
always has a half-integral solution $x^*\in \left\{0,\frac12,1\right\}^{|E^\dagger|}$ such that
the collection of its half-integral edges forms disjoint odd cycles.
\end{theorem}
\begin{proof}
For the proof of Theorem \ref{thm:halfintegral},
once we show the half-integrality of LP \eqref{lp:modifiededmonds}, it is easy to check that the half-integral edges forms disjoint odd cycles.
Hence, it suffices to show that every vertex of the polytope consisting of constraints of LP \eqref{lp:modifiededmonds} is always half-integral.
To this end, we use the following lemma {which is proven in the appendix}.
\begin{lemma}\label{lem:halfinv}
Let $A=[A_{ij}]\in\{0,1\}^{m\times m}$ be an invertible $0$-$1$ matrix whose row has at most two non-zero entires.
Then, each entry $A^{-1}_{ij}$ of $A^{-1}$ is in $\left\{0,\pm 1,\pm\frac{1}{2}\right\}$.
\end{lemma}
Consider a vertex $x\in [0,1]^{|E^\dagger|}$ of the polytope consisting of constraints of LP \eqref{lp:modifiededmonds}.
Then, there exists a linear system of equalities such that
$x$ is its unique solution where each equality is either $x_e=0$,
$x_e=1$ or $\sum_{e\in \delta(v)} x_e=1$.
One can plug $x_e=0$ and $x_e=1$ into the linear system, reducing it 
to $Ax=b$ where $A$ is an invertible $0$-$1$ matrix whose column contains at most two non-zero entries.
Hence, from Lemma \ref{lem:halfinv}, $x$ is half-integral. This completes the proof of
Theorem \ref{thm:halfintegral}.
\end{proof}
Next let us design BP for obtaining the half-integral solution of LP \eqref{lp:modifiededmonds}.
First, we duplicate each edge $e\in E^\dagger$ into $e_1,e_2$ and define a new graph $G^\ddagger=(V^\dagger,E^\ddagger)$ where $E^\ddagger=\{e_1,e_2:e\in E^\ddagger\}$.
Then, we build the following equivalent LP:
\begin{equation}\label{lp:modifiededmonds-1}
\begin{split}
\mbox{minimize}\qquad&w^{\ddagger}\cdot x\\
\mbox{subject to}\qquad&\sum_{e\in\delta(v)}x_e=2,\quad\forall~ v\in V^{\dagger},~\text{$v$ is a non-blossom vertex}\\
&\sum_{e\in\delta(v)}x_e\ge 2,\quad\forall~ v\in V^{\dagger},~\text{$v$ is a blossom vertex}\\
&x=[x_e]\in[0,1]^{|E^{\dagger}|},
\end{split}
\end{equation}
where $w^\ddagger_{e_1}=w^\ddagger_{e_2}=w_e^\dagger$. 
One can easily observe that solving LP \eqref{lp:modifiededmonds-1} is equivalent to solving LP \eqref{lp:modifiededmonds}
due to our construction of $G^\ddagger,w^\ddagger$, and 
LP \eqref{lp:modifiededmonds-1} always have an integral solution due to Theorem \ref{thm:halfintegral}.
Now, construct the following GM for LP \eqref{lp:modifiededmonds-1}: 
\begin{equation}\label{gm:modmatching}
	\Pr[X=x]~\propto~\prod_{e\in E^\ddagger} e^{w^\ddagger_{e} x_{e}}\prod_{v\in V^\dagger} \psi_{v} (x_{\delta(v)}),
\end{equation}
where the factor function $\psi_v$ is defined as
\begin{align*}
&\psi_{v}(x_{\delta(v)}) = 
\begin{cases}
1&\mbox{if $v$ is a non-blossom vertex and}~ \sum_{e \in\delta(v)} x_{e}= 2\\
1&\mbox{else if $v$ is a blossom vertex and}~ \sum_{e \in\delta(v)} x_{e}\geq 2\\
0&\mbox{otherwise}
\end{cases}.
\end{align*}
For this GM,
we derive the following corollary of Theorem \ref{thm:bplp} {proven in the
appendix.}
\begin{corollary}\label{cor:matching}
If LP \eqref{lp:modifiededmonds-1} has a unique solution, then 
the max-product BP applied to GM \eqref{gm:modmatching} converges to it. 
\end{corollary}
The uniqueness condition stated in the corollary above is easy to guarantee by adding small random noise
corrections to edge weights.
Corollary \ref{cor:matching} shows that BP can compute the half-integral solution of LP \eqref{lp:modifiededmonds}.

\vspace{-0.1in}
\section{Proof of Theorem \ref{thm:main}}\label{sec:proof}

First, it is relatively easy to prove the correctness of Blossom-BP, as stated in the following lemma.
\begin{lemma}\label{lem:optimality}
If Blossom-LP terminates, 
it outputs the minimum weight perfect matching. 
\end{lemma}
\begin{proof}
We 
let
$x^\dagger=[x_e^\dagger],y^\ddagger=[y^\ddagger_v,y^\ddagger_S:v\notin V^\dagger, v(S)\notin V^\dagger]$
denote the parameter values at the termination of Blossom-BP.
Then, the strong duality theorem and the complementary slackness condition imply that
\begin{equation}
x_e^\dagger(w^\dagger - y_u^\dagger-y_v^\dagger)=0,\quad \forall e=(u,v)\in E^\dagger.\label{eq:compslack}
\end{equation}
where $y^{\dagger}$ be a dual solution of $x^{\dagger}$.
Here, observe that $y^\dagger$ and $y^\ddagger$ cover $y$-variables inside and outside of $V^\dagger$, respectively.
Hence, one can naturally define $y^{*} = [y_{v}^{\dagger}~y_{u}^{\ddagger}]$ to cover all $y$-variables, i.e., $y_v,y_S$
for all $v\in V, S\in \mathcal L$.
If we define $x^{*}$ for the output matching $M^*$ of Blossom-LP as
$x^{*}_{e} =1$ if $e\in M^*$ and $x^*_e=0$ otherwise,
then $x^{*}$ and $y^{*}$ satisfy the following complementary slackness condition:
\begin{align*}
&x^{*}_e\left(w_e - y^{*}_{u} - y^{*}_{v} - \sum_{S\in\mathcal L}y^{*}_{S}\right)=0,\quad \forall e=(u,v)\in E,\qquad
y^{*}_S\left(\sum_{e\in\delta(S)}x^{*}_e-1\right)=0,\quad \forall S\in\mathcal L,
\end{align*}
where $\mathcal L$ is the last set of blossoms at the termination of Blossom-BP.
In the above, 
the first equality is from \eqref{eq:compslack} and the definition of $w^\dagger$,
and the second equality is because the construction of $M^*$ in Blossom-BP is designed to enforce $\sum_{e\in\delta(S)}x^{*}_e=1$.
This proves that $x^{*}$ is the optimal solution of LP \eqref{lp:edmonds} 
and $M^{*}$ is the minimum weight perfect matching, thus completing the proof of Lemma \ref{lem:optimality}.
\end{proof}

To guarantee the termination of Blossom-LP in polynomial time,  
we use the following notions.

\begin{definition}\label{def:claw}
Claw is a subset of edges such that every edge in it 
shares a common vertex, called center, with all other edges,
i.e., the claw forms a star graph.
\end{definition}
\begin{definition}\label{def:CLM}
Given a graph $G=(V,E)$, a set of odd cycles $\mathcal O\subset 2^E$, a set of claws $\mathcal W\subset 2^E$
and a matching $M\subset E$,
$(\mathcal O, \mathcal W, M)$ is called cycle-claw-matching decomposition of $G$ if
all sets in $\mathcal O\cup \mathcal W\cup \{M\}$ are disjoint and
each vertex $v\in V$ is covered by exactly one set among them.
\end{definition}
To analyze the running time of Blossom-BP,
we construct an iterative auxiliary algorithm that outputs 
the minimum weight perfect matching in a bounded number of iterations. 
The auxiliary algorithm outputs a cycle-claw-matching decomposition at each iteration,
and it terminates when the cycle-claw-matching decomposition corresponds to a perfect matching.
We will prove later that 
the auxiliary algorithm and Blossom-LP
are equivalent and, therefore, conclude that the iteration of Blossom-LP is also bounded.

To design the auxiliary algorithm, we consider the following dual of LP \eqref{lp:modifiededmonds}:
\begin{equation}\label{lp:modifiededmondsdual}
\begin{split}
\mbox{minimize}\qquad& \sum_{v\in V^\dagger} y_v\\
\mbox{subject to}\qquad&w^{\dagger}_e-y_v - y_u\geq 0,\quad\forall e=(u,v)\in E^\dagger,
\qquad y_{v(S)}\geq0,\quad \forall S\in \mathcal L.
\end{split}
\end{equation}
Next we introduce an auxiliary iterative algorithm which updates iteratively 
the blossom set $\mathcal L$ and also
the set of variables $y_{v}, y_S$ for $v\in V, S\in \mathcal L$.
We call edge $e=(u,v)$ `tight' if
$$w_{e} - y_{u} - y_{v}-\sum_{S\in\mathcal L:e\in\delta (S)}y_{S} = 0.$$
Now, we are ready to describe the auxiliary algorithm having the following parameters.
\begin{itemize}
\item[$\circ$] $G^{\dagger}=(V^\dagger, E^\dagger)$,
$\mathcal L\subset 2^V$, 
and $y_v, y_S$ for $v\in V, S\in \mathcal L$.
\item[$\circ$]  $(\mathcal O, \mathcal W, M)$: 
A cycle-claw-matching decomposition of $G^{\dagger}$
\item[$\circ$]  $T\subset G^\dagger$: A tree graph 
consisting of $+$ and $-$ vertices.
\end{itemize}
Initially, set $G^{\dagger}=G$ and $\mathcal L, T=\emptyset$. 
In addition, set $y_v, y_S$ by an optimal solution of LP \eqref{lp:modifiededmondsdual} with $w^{\dagger} = w$
and $(\mathcal O, \mathcal W, M)$ by 
the cycle-claw-matching decomposition of $G^{\dagger}$
consisting of tight edges with respect to $[y_v,y_S]$. The parameters are updated iteratively as follows.

\vspace{0.15in}

\noindent {\bf The auxiliary algorithm}
\hrule
\paragraph{Iterate the following steps until $M$ becomes a perfect matching:}
\begin{itemize}
\item[1.] Choose a vertex $r\in V^\dagger$ from the following rule.
\begin{itemize}
\item[] {\bf Expansion.} If $\mathcal W\neq\emptyset$, choose a claw $W\in\mathcal W$ of center blossom vertex $c$ and
choose a non-center vertex $r$ in $W$. 
Remove the blossom $S(c)$ corresponding to $c$ from $\mathcal L$ and update $G^{\dagger}$ by expanding it.
Find a matching $M^\prime$ covering all vertices in $W$ and $S(c)$ except for $r$
and update 
$M \leftarrow M\cup M^\prime$.
\vspace{0.05in}
\item[] {\bf Contraction.} Otherwise, 
choose a cycle $C \in \mathcal O$, add and remove it from $\mathcal L$ and $\mathcal O$, respectively.
In addition, $G^{\dagger}$ is also updated by contracting $C$ and
choose the contracted vertex $r$ in $G^\dagger$ and set $y_r=0$.
\end{itemize} 
Set tree graph $T$ having $r$ as $+$ vertex and no edge.
\item[2.] Continuously increase $y_{v}$ of every $+$ vertex $v$  in $T$
and decrease $y_{v}$ of $-$ vertex $v$ in $T$ by the same amount
until one of the following events occur:
\begin{itemize}
\item[] {\bf Grow.}
If a tight edge $ (u, v)$ exists where $u$ is a $+$ vertex of $T$ and $v$ is covered by $M$, 
find a tight edge $(v, w) \in M$. Add edges $(u,v),(v,w)$ to $T$ and remove $(v,w)$ from $M$
where $v,w$ becomes $-,+$ vertices of $T$, respectively.
\vspace{0.05in}
\item[] {\bf Matching.} 
If a tight edge $ (u, v)$ exists where $u$ is a $+$ vertex of $T$ and $v$ is covered by $C\in \mathcal O$, 
find a matching $M^\prime$ that covers $T\cup C$.
Update $M \leftarrow M\cup M^\prime$ 
and remove $C$ from $\mathcal O$.
\vspace{0.05in}
\item[] {\bf Cycle.}
If a tight edge $(u, v)$ exists where $u,v$ are $+$ vertices of $T$,
find a cycle $C$ and a matching $M^\prime$ that covers $T$.
Update $M\leftarrow M\cup M^\prime$ and
add $C$ to $\mathcal O$.
\vspace{0.05in}
\item[] {\bf Claw.}
If a blossom vertex $v(S)$ with $y_{v(S)}=0$ exists,
find a claw $W$ (of center $v(S)$) and a matching $M^\prime$ covering $T$.
Update $M \leftarrow M\cup M^\prime$ and  add
$W$ to $\mathcal W$. 
\end{itemize}
If {\bf Grow} occurs, resume the step 2. Otherwise, go to the step 1.
\end{itemize}
\vspace{0.1in}
\hrule
\vspace{0.1in}
\begin{figure}[h]
\centering
\subfigure[Grow]{
	\includegraphics[scale=0.4]{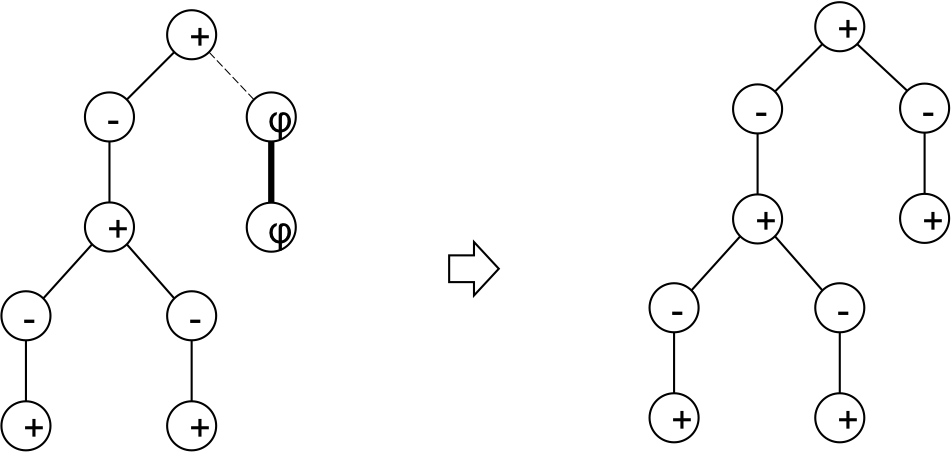}
	\label{fig:aux_grow}
}
\subfigure[Matching]{
	\includegraphics[scale=0.4]{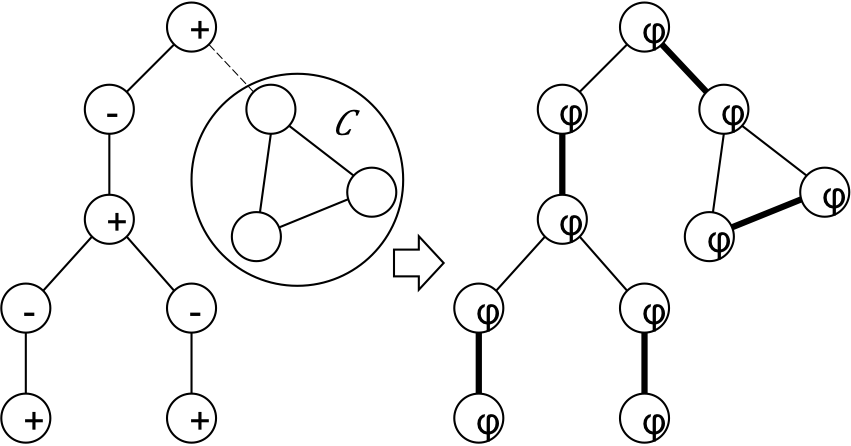}
	\label{fig:aux_aug}
}
\subfigure[Cycle]{
	\includegraphics[scale=0.4]{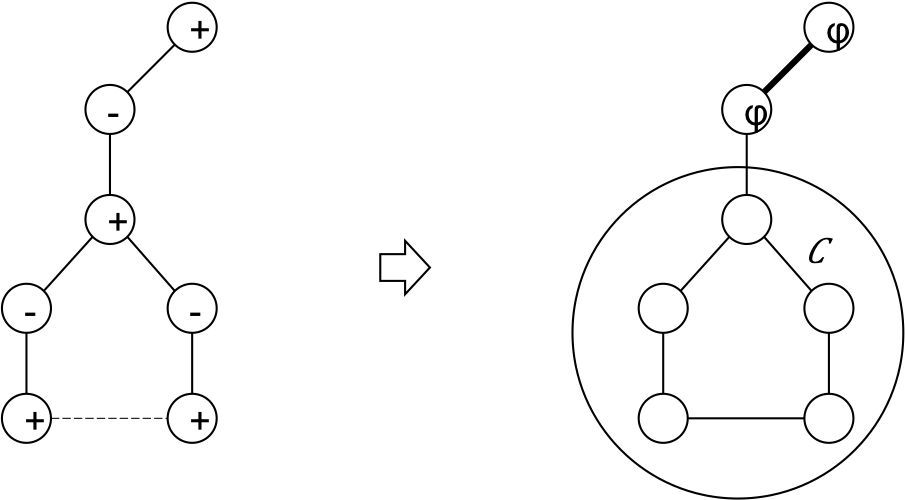}
	\label{fig:aux_shrink}
}
\subfigure[Claw]{
	\includegraphics[scale=0.4]{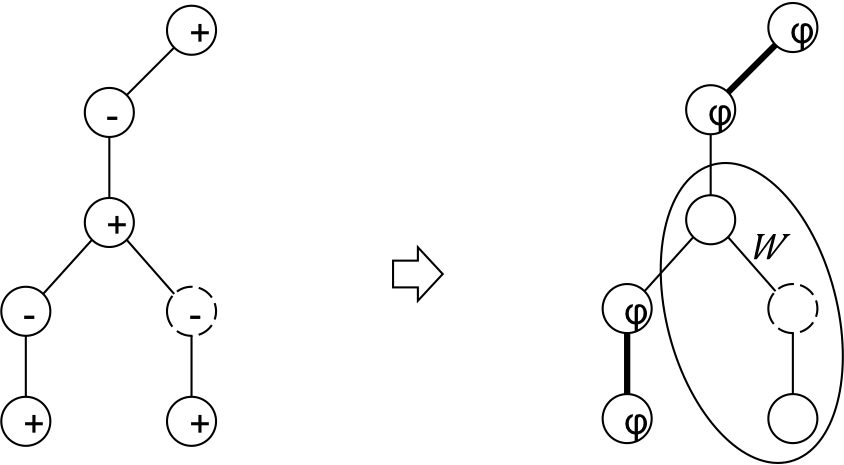}
	\label{fig:aux_expand}
}
\caption{Illustration of four possible executions of Step 2 of the auxiliary algorithm.
Here, we use $\varphi$ to label vertices covered by matching $M$ appearing at the intermediate steps of 
the auxiliary algorithm.}
\end{figure}
Note that the auxiliary algorithm updates parameters in such a way that 
the number of vertices in every claw in the cycle-claw-matching decomposition is $3$ since every
$-$ vertex has degree $2$.
Hence, there exists a unique matching $M^\prime$ in the expansion step.
Furthermore, the existence of a cycle-claw-matching decomposition
at the initialization can be guaranteed
using the complementary slackness condition and the half-integrality
of LP \eqref{lp:modifiededmonds}.
We establish the following lemma for the running time of the auxiliary algorithm.\\
\begin{lemma}\label{lem:auxrunningtime}
The auxiliary algorithm terminates in 
$O(|V|^{2})$ iterations.
\end{lemma}
\begin{proof}
To this end, 
let $(\mathcal O, \mathcal W, M)$ be the cycle-claw-matching decomposition of $G^{\dagger}$
and $N = |\mathcal O| + |\mathcal W|$ at some iteration of the algorithm. 
We first prove that $|\mathcal O| + |\mathcal W|$ does not increase at every iteration.
At Step 1, the algorithm deletes an element in either $\mathcal O$ or $\mathcal W$ 
and hence, 
$|\mathcal O| + |\mathcal W| = N - 1$.
On the other hand,
at Step 2, one can observe that the algorithm run into one of the following scenarios with respect to $|\mathcal O| + |\mathcal W|$:
\begin{itemize}
\item[]{\bf Grow.} $|\mathcal O| + |\mathcal W| = N -1$
\item[]{\bf Matching.}  $|\mathcal O| + |\mathcal W| = N-2$
\item[]{\bf Cycle.} $|\mathcal O| + |\mathcal W| = N$
\item[]{\bf Claw.} $|\mathcal O| + |\mathcal W| = N$
\end{itemize}
Therefore, the total number of odd cycles and claws at Step 2 does not increase as well.

From now on, we define $\{t_1,t_2,\dots:t_i\in \mathbb Z\}$ to be
indexes of iterations when {\bf Matching} occurs at Step 2, and
we call the set of iterations $\{t:t_i\leq t< t_{i+1}\}$ as the $i$-th {\em stage}.
We will show that the length of each stage is $O(|V|)$, i.e., for all $i$,
\begin{equation}
|t_i-t_{i+1}|=O(|V|).\label{eq:lenstage}
\end{equation} 
This implies
that the auxiliary algorithm terminates in $O(|V|^2)$ iterations
since the total number of odd cycles and claws at the initialization is $O(|V|)$
and it decrease by two if {\bf Matching} occurs.
To this end, we prove the following key lemmas{, which are proven in the appendix}.
\begin{claim}\label{claim1}
At every iteration of the auxiliary algorithm,
there exist no path consisting of tight edges between two vertices $v_1,v_2\in V^\dagger$  
where each $v_i$ is either a blossom vertex $v(S)$ with $y_{S} = 0$ or a (blossom or non-blossom) vertex in an odd cycle
consisted of tight edges. 
\end{claim}
\begin{claim}\label{claim2}
Consider a $+$ vertex $v\in V^\dagger$ at some iteration of the  auxiliary algorithm.
Then,  
at the first iteration afterward where $v$ becomes a $-$ vertex or is removed from $V^\dagger$ (i.e., due to the contraction of a blossom), it 
is connected 
to an odd cycle $C\in \mathcal O$ 
via an even-sized alternating  path consisting of tight edges with respect to matching $M$
whenever each iteration starts during the same stage.
Here, $\mathcal O$ and $M$ are from the cycle-claw-decomposition.
\end{claim}
Now we aim for proving \eqref{eq:lenstage}.
To this end, 
we 
claim the following.
\begin{itemize}
\item[$\spadesuit$]
A $+$ vertex of $V^\dagger$ at some iteration cannot be a $-$ one (whenever it appears in $V^\dagger$) afterward in the same stage.
\end{itemize}
For proving $\spadesuit$,
we assume that a $+$ vertex $v\in V^\dagger$ at the $t$-th iteration violates $\spadesuit$ to derive a contradiction, i.e., 
it becomes a $-$ one in some tree $T$ during $t^\prime$-th iteration in the same stage.
Without loss of generality, one can assume that
the vertex $v$ has the minimum value of $t^\prime - t$ among such vertices violating $\spadesuit$.
We consider two cases: (a) $v$ is always contained in $V^\dagger$ afterward in the same stage,
and (b) $v$ is removed from $V^{\dagger}$ (at least once, due to the contraction of a blossom containing $v$) 
afterward in the same stage.
First consider the case (a).
Then, due to the assumption of the case (a) and Claim \ref{claim2}, 
there exist a path $P$ from $v$ to a cycle $C\in\mathcal O$  when the $t^\prime$-th iteration starts.
Then, one can observe that 
in order to add $v$ to tree $T$ as a $-$ vertex,
it must be the first vertex in path $P$ added to $T$ by {\bf Grow} during the $t$-iteration.
Furthermore,
tree $T$ keeps continuing 
to perform {\bf Grow} afterward using tight edges of path $P$
without modifying parameter $y$ until {\bf Matching} occurs, i.e., the new stage starts.
This is because {\bf Claw} and {\bf Cycle} are impossible to occur before {\bf Matching} due
to Claim \ref{claim1}.
Hence, it contradicts to the assumption that $t$ and $t^\prime$ are in the same stage,
and completes the proof of $\spadesuit$ for the case (a).
Now we consider the case (b), i.e., $v$ is removed from $V^{\dagger}$
due to the contraction of a blossom $S\in \mathcal L$.
In this case, the blossom vertex $v(S)\in V^{\dagger}$ must be expanded before $v$ becomes a $-$ vertex.
However, $v(S)$ becomes a $+$ vertex after contracting $S$ 
and a $-$ vertex before expanding $v(S)$, i.e., $v(S)$ also violates $\spadesuit$.
This contradicts to the assumption that
the vertex $v$ has the minimum value of $t^\prime - t$ among vertices violating $\spadesuit$,
and completes the proof of $\spadesuit$.
Due to $\spadesuit$, a blossom cannot expand after contraction in the same stage, where
we remind that a blossom vertex becomes a $+$ one after contraction and 
a $-$ one before expansion.
This implies that 
the number contractions and expansions in the same stage is $O(|V|)$, which
leads to \eqref{eq:lenstage} and completes the proof of Lemma \ref{lem:auxrunningtime}. 
\end{proof}
Now we are ready to prove the equivalence between the auxiliary algorithm and the
Blossom-LP, 
i.e., prove that the numbers
 of iterations of Blossom-LP and 
 the auxiliary algorithm are equal.
To this end, given a cycle-claw-matching decomposition $(\mathcal O,\mathcal W,M)$,
observe that one can choose the corresponding $x=[x_e]\in \{0,1/2,1\}^{|E^\dagger|}$
that satisfies constraints of LP \eqref{lp:modifiededmonds}:
$$x_e=\begin{cases}
1&\mbox{if $e$ is an edge in $\mathcal W$ or $M$}\\
\frac{1}{2}&\mbox{if $e$ is an edge in $\mathcal O$}\\
0&\mbox{otherwise}
\end{cases}.$$
Similarly, given a half-integral $x=[x_e]\in \{0,1/2,1\}^{|E^\dagger|}$
that satisfies constraints of LP \eqref{lp:modifiededmonds}, one can find the corresponding 
cycle-claw-matching decomposition.
Furthermore, one can also define weight $w^\dagger$ in $G^\dagger$ for the auxiliary algorithm 
as Blossom-LP does:
\begin{equation}\label{eq:wdagger}
w_{e}^{\dagger} =  w_e-
\sum_{v\in V: v\not\in V^{\dagger},e\in \delta (v)}y_v
-\sum_{S\in \mathcal L: v(S)\not\in V^{\dagger},e\in \delta (S)}y_S, 
\qquad \forall~e\in E^{\dagger}.
\end{equation}
In the auxiliary algorithm, $e=(u,v)\in E^\dagger$ is tight if and only if
$$w^{\dagger}_e - y^{\dagger}_{u} - y^{\dagger}_{v}=0.$$
Under these equivalences in parameters between Blossom-LP and the auxiliary algorithm, 
we will use the induction to show that cycle-claw-matching decompositions
maintained by both algorithms are equal at every iteration, as stated in the following lemma.

\begin{lemma}\label{lem:equivalence}
Define the following notation:
\begin{align*}
y^{\dagger} = [y_{v}: v \in V^{\dagger}]\qquad\mbox{and}\qquad
y^{\ddagger} = [y_{v},y_S: v\in V,v \not\in V^{\dagger}, S\in \mathcal L, v(S)\notin V^\dagger],
\end{align*}
i.e., $y^\dagger$ and $y^{\ddagger}$ are parts of $y$ which involves and does not involve in $V^\dagger$, respectively.
Then, the Blossom-LP and the auxiliary algorithm update parameters
$\mathcal L, y^{\ddagger}$ equivalently and output the
same cycle-claw-decomposition of $G^{\dagger}$ at each iteration.
\end{lemma}
\begin{proof}
Initially, it is trivial. 
Now we assume the induction hypothesis that $\mathcal L, y^{\ddagger}$ and the cycle-claw-decomposition
are equivalent between both algorithms at the previous iteration.
First, it is easy to observe that
$\mathcal L$ is updated equivalently since it is only decided by 
the cycle-claw-decomposition at the previous iteration in both algorithms.
Next, it is also easy to check that $y^{\ddagger}$ is updated equivalently since
(a) if we remove a blossom $S$ from $\mathcal L$, it is trivial
and (b) if we add a blossom $S=V(C)$ for some cycle $C$ to $\mathcal L$, 
$y^{\ddagger}$ is uniquely decided by $C$ and $w^\dagger$ in both algorithms.

In the remaining of this section, we will show that
once $\mathcal L, y^\ddagger$ are updated equivalently,
the cycle-claw-decomposition also changes equivalently in both algorithms. Observe that $G^\dagger, w^\dagger$
only depends on $\mathcal L, y^\ddagger$.
In addition, $y^{\dagger}$ maintained by the auxiliary algorithm also satisfies
constraints of LP \eqref{lp:modifiededmondsdual}.
Consider the cycle-claw-matching decomposition $(\mathcal O,\mathcal W,M)$ of the auxiliary algorithm,
and the corresponding $x=[x_e]\in \{0,1/2,1\}^{|E^\dagger|}$ that satisfies constraints of
LP \eqref{lp:modifiededmonds}. 
Then, $x$ and $y^{\dagger}$ satisfy 
the complementary slackness condition:
\begin{align*}
&x_e(w^{\dagger}_e - y^{\dagger}_{u} - y^{\dagger}_{v})=0,~\qquad \forall e=(u,v)\in E^{\dagger}\\
&y^{\dagger}_{v(S)}\left(\sum_{e\in\delta(v(S))}x_e-1\right)=0,~\qquad \forall S\in\mathcal L,
\end{align*}
where the first equality is because 
the cycle-claw-matching decomposition consists of tight edges
and the second equality is because 
every claw maintained by the auxiliary algorithm
has its center vertex $v(S)$ with
$y_{v(S)}=0$ for some $S\in \mathcal L$.
Therefore, $x$ is an optimal solution of LP \eqref{lp:modifiededmonds}, i.e., 
the cycle-claw-decomposition is updated equivalently in both algorithms.
This completes the proof of Lemma \ref{lem:equivalence}.
\end{proof}
The above lemma implies that 
Blossom-LP also terminates in $O(|V|^{2})$ iterations due to Lemma \ref{lem:auxrunningtime}.
This completes the proof of Theorem \ref{thm:main}.
The equivalence between the half-integral solution of LP \eqref{lp:modifiededmonds}
in Blossom-LP and
the cycle-claw-matching
decomposition in the auxiliary algorithm
implies that LP \eqref{lp:modifiededmonds} is always has a half-integral solution, and hence,
one of Steps {\bf B.}(a), {\bf B.}(b) or {\bf B.}(c) always occurs.

\vspace{-0.05in}
\section{Conclusion}
The BP algorithm has been popular for approximating inference solutions arising in graphical
models, where its distributed implementation, associated ease of programming and strong parallelization
potential are the main reasons for its growing popularity. This paper aims for designing a polynomial-time
BP-based scheme solving the minimum weigh perfect matching problem.
We believe that our approach is of a
broader interest to advance the challenge of designing 
BP-based MAP solvers in more general GMs as well as
distributed (and parallel) solvers for large-scale IPs.

{
\begingroup
\subsubsection*{References}
\renewcommand{\section}[2]{}

\endgroup}

\appendix

\section{Proof of Lemma \ref{lem:halfinv}}
For the proof of Lemma \ref{lem:halfinv}, suppose there exists a row in $A$ with one non-zero entry.
Then, one can assume that it is the first row of $A$ and $A_{11}=1$ without loss of generality.
Hence, $A^{-1}_{11}=1$, $A^{-1}_{1i}=0$ for $i\ne 1$ and
the first column of $A^{-1}$ has only $0$ and $\pm 1$ entries since each row of $A$
has at most two non-zero entries.
This means that one can proceed the proof of Lemma \ref{lem:halfinv} for the submatrix of $A$
deleting the first row and column.
Therefore, one can assume that
each row of $A$ contains exactly two non-zero entries.

We construct a graph $G=(V,E)$ such that
\begin{align*}
&V=[m]:=\{1,2,\dots,m\}\qquad\mbox{and}\qquad E=\{(j,k):a_{ij}=a_{ik}=1~\text{for some}~i\in V\},
\end{align*}
i.e., each row $A_{i[m]}=(A_{i1},\dots,A_{im})$ and
each column 
$A_{[m]i}=(A_{1i},\dots,A_{mi})^T$ correspond to an edge and a vertex of $G$, respectively.
Since $A$ is invertible, 
one can notice that $G$ does not contain an even cycle as well as
a path between two distinct odd cycles (including two odd cycles share a vertex). 
Therefore, each connected component of $G$ has at most one odd cycle.
Consider the $i$-th column $A^{-1}_{[m]i}=(A^{-1}_{1i},\dots,A^{-1}_{mi})^T$ of $A^{-1}$ and we have
\begin{equation}\label{eq:inv}
A_{i[m]}A^{-1}_{[m]i}=1\qquad\text{and}\qquad A_{j[m]}A^{-1}_{[m]i}=0\quad \text{for}~j\ne i,
\end{equation}
i.e., $A^{-1}_{[m]i}$ assigns some values on $V$ such that 
the sum of values on two end-vertices of the edge 
corresponding to the $k$-th row of $A$
is $1$ and $0$ if $k=i$ and $k\neq i$, respectively. 

Let $e=(u,v)\in E$ be the edge corresponding 
to the $i$-th row of $A$.
\begin{itemize}
\item First, consider the case when $e$ is not in an odd cycle of $G$.
Since each component of $G$ contains at most one odd cycle, 
one can assume that
the component of $u$ is a tree in the graph $G\setminus e$.
We will find the entries of $A^{-1}$ satisfying \eqref{eq:inv}.
Choose $A^{-1}_{wi}=0$ for all vertex $w$ not 
in the component. 
and $A^{-1}_{ui}=1$.
Since the component forms a tree, 
one can set $A^{-1}_{wi}=1~\text{or}~-1$ for every vertex $w\neq u$ in the component to satisfy \eqref{eq:inv}.
This implies that $A^{-1}_{[m]i}$ consists of $0$ and $\pm 1$.
\item Second, consider the case when $e$ is in an odd cycle of $G$.
We will again find the entries of $A^{-1}$ satisfying \eqref{eq:inv}.
Choose $A^{-1}_{ui}=A^{-1}_{vi}=\frac{1}{2}$ and $A^{-1}_{wi}=0$ for every
vertex $w$ not in the component containing $e$. 
Then, one can choose $A^{-1}_{[m]i}$ satisfying \eqref{eq:inv} by assigning 
$A^{-1}_{wi}=\frac{1}{2}~\text{or}~-\frac{1}{2}$ for vertex $w\neq u,v$ in the component containing $e$.
Therefore, $A^{-1}_{[m]i}$ consists of $0$ and $\pm\frac{1}{2}$.
\end{itemize}
This completes the proof of Lemma \ref{lem:halfinv}.

\section{Proof of Corollary \ref{cor:matching}}\label{sec:pf:cor:matching}
The proof of Corollary \ref{cor:matching} will be completed using Theorem \ref{thm:bplp}.
If LP \eqref{lp:modifiededmonds-1} has a unique solution,
LP \eqref{lp:modifiededmonds-1} has a unique and integral solution by Theorem \ref{thm:halfintegral},
 i.e., Condition {\it C1} of Theorem \ref{thm:bplp}.
LP \eqref{lp:modifiededmonds-1} satisfies Condition {\it C2} as each edge is incident with two vertices. 
Now, we need to prove that LP \eqref{lp:modifiededmonds-1} satisfies Condition {\it C3} of Theorem \ref{thm:bplp}.
Let $x^*$ be a unique optimal solution of LP \eqref{lp:modifiededmonds-1}.
Suppose $v$ is a non-blossom vertex and $\psi_v(x_{\delta(v)})=1$ for some 
$x_{\delta(v)}\ne x^*_{\delta(v)}$.
If $x_e\ne x^*_e=1$ for $e\in\delta(v)$, there exist $f\in\delta(v)$ such that
$x_f\ne x^*_f=0$. Similarly, If $x_e\ne x^*_e=0$ for $e\in\delta(v)$, there exists 
$f\in\delta(v)$ such that $x_f\ne x^*_f=1$. Then, it follows that
$$
\psi_v(x^\prime_\delta(v))=1,\qquad
\mbox{
where 
$x^\prime_{e^\prime} = \begin{cases}
		x_{e^\prime}~&\mbox{if}~e^\prime\notin \{e,f\}\\
		x^*_{e^\prime}~&\mbox{otherwise}
\end{cases}$.}$$
$$
\psi_v(x^\prime_\delta(v))=1,\qquad
\mbox{
where 
$x^\prime_{e^\prime} = \begin{cases}
		x_{e^\prime}~&\mbox{if}~e^\prime\in \{e,f\}\\
		x^*_{e^\prime}~&\mbox{otherwise}
\end{cases}$.}$$
Suppose $v$ is a blossom vertex and $\psi_v(x_{\delta(v)})=1$ for some 
$x_{\delta(v)}\ne x^*_{\delta(v)}$.
If $x_e\ne x^*_e=1$ for $e\in\delta(v)$, choose $f\in\delta(v)$ such that
$x_f\ne x^*_f=0$ if it exists. Otherwise, choose $f=e$.
Similarly, If $x_e\ne x^*_e=0$ for $e\in\delta(v)$, choose 
$f\in\delta(v)$ such that $x_f\ne x^*_f=1$ if it exists.
Otherwise, choose $f=e$. Then, it follows that
$$
\psi_v(x^\prime_\delta(v))=1,\qquad
\mbox{
where 
$x^\prime_{e^\prime} = \begin{cases}
		x_{e^\prime}~&\mbox{if}~e^\prime\notin \{e,f\}\\
		x^*_{e^\prime}~&\mbox{otherwise}
\end{cases}$.}$$
$$
\psi_v(x^\prime_\delta(v))=1,\qquad
\mbox{
where 
$x^\prime_{e^\prime} = \begin{cases}
		x_{e^\prime}~&\mbox{if}~e^\prime\in \{e,f\}\\
		x^*_{e^\prime}~&\mbox{otherwise}
\end{cases}$.}$$

\section{Proof of Claim \ref{claim1}}
First observe that $w^{\dagger}$ (see \eqref{eq:wdagger} for its definition)
is updated only at {\bf Contraction} and {\bf Expansion} of Step 1.
If {\bf Contraction} occurs, 
there exist a cycle $C$ to be contracted before Step 1.
Then one can observe that before the contraction, for every vertex $v$ in $C$, 
$y_{v}$ is expressed as a linear combination of $w^{\dagger}$:
\begin{equation}
y_{v} = {1\over 2}\sum_{e\in E(C)}(-1)^{d_{C}(e,v)}w^{\dagger}_{e},\label{eq:tightcycle}
\end{equation}
where $d_{C}(v,e)$ is the graph distance from vertex $v$ to edge $e$ in the odd cycle $C$.
Moreover $w^{\dagger}$ is updated after the contraction as
\begin{equation*}
\begin{cases}
w^{\dagger}_e \leftarrow w^{\dagger}_e - y_{v} \qquad \text{if $v$ is in the cycle $C$ and}~ e\in \delta(v)\\
w^{\dagger}_e \leftarrow w^{\dagger}_e \qquad\qquad \text{otherwise}
\end{cases}.
\end{equation*}
Thus the updated value $w^{\dagger}_e$ can be expressed as a linear combination of the old values $w^{\dagger}$ 
where each coefficient is uniquely determined by $G^\dagger$.
One can show the same conclusion similarly when
{\bf Expansion} occurs.
Therefore one conclude the following.
\begin{itemize}
\item[$\clubsuit$] Each value $w^{\dagger}_e$ at any iteration can be expressed as a linear combination of 
the original weight values $w$ where each coefficient is uniquely determined by the prior history in $G^\dagger$.
\end{itemize}


To derive a contradiction, we assume there exist a path $P$ 
consisting of tight edges
between two vertices $v_{1}$ and $v_{2}$ where each $v_{i}$ is either a blossom vertex $v(S)$ with
$y_S = 0$ or a vertex in an odd cycle consisting of tight edges. 
Consider the case where $v_1$ and $v_2$ are in cycle $C_1$ and $C_2$ consisting of tight edges, where
other cases can be argued similarly.
Then one can observe that there exists a linear relationship between $y_{v}$ and $y_{u}$ and $w^{\dagger}$:
\begin{equation}\label{eq:tightpath}
y_{v_1} + (-1)^{d_{P}(v_2, v_1)}y_{v_2} = \sum_{e \in P}(-1)^{d_{P}(e,v_1)}w_{e}^{\dagger}
\end{equation}
where $d_{P}(v_2, v_1)$ and $d_{P}(e, v_1)$ is the graph distance from $v_1$ to $v_2$ and $e$, respectively,
in the path $P$.
Since $v_1, v_2$ are in cycles $C_1,C_2$, respectively, we can apply 
\eqref{eq:tightcycle}.
From this observation, \eqref{eq:tightpath} and $\clubsuit$,
there exists a linear relationship among the original weight values $w$,
where each coefficient is uniquely determined by the prior history in $G^\dagger$.
This is impossible since
the number of possible scenarios in the history of $G^\dagger$ is finite, whereas
we add continuous random noises to $w$. 
This completes the proof of Claim \ref{claim1}.

\section{Proof of Claim \ref{claim2}}
To this end, suppose that a $+$ vertex $v$ at the $t^\dagger$-th iteration 
first becomes a $-$ vertex or is removed from $V^\dagger$ at the $t^\ddagger$-th iteration
where $t^\dagger$, $t^\ddagger$-th iterations are in the same stage.
First observe that if $v$ is removed from $G^{\dagger}$ at the $t^\ddagger$-th iteration,
there exist a cycle in $\mathcal O$ that includes it at the start of the $t^\ddagger$-th iteration, 
resulting a zero-sized alternating path
between such vertex and cycle, i.e., the conclusion of Lemma \ref{claim2} holds.
Now, for the other case, i.e., $v$ becomes a $-$ vertex at the $t^\ddagger$-th iteration, we will prove the following.
\begin{itemize}
\item[$\bigstar$]
For any $t$-th iteration with $t^\dagger\leq t<t^\ddagger$, one of the followings holds:
\begin{itemize}
\item[1.] The vertex $v$ becomes a $+$ vertex during the $t$-th iteration.
Moreover, $v$ either becomes a $+$ vertex during the $(t+1)$-th iteration or
$v$ becomes connected to some cycle $C$ in $\mathcal O$ 
via an even-sized alternating path $P$
consisting of tight edges
at the start of $(t+1)$-th iteration.
\item[2.] The vertex $v$ is not in the tree $T$ during the $t$-th iteration.
Moreover, if $v$ is connected to some cycle $C$ in $\mathcal O$ 
via an even-sized alternating path $P$
consisting of tight edges
at the start of $t$-th iteration, 
$v$ remains connected to cycle $C$  in $\mathcal O$ via 
an even-sized alternating path $P$ consisted of tight edges 
at the start of $(t+1)$-th iteration, i.e. the algorithm parameters associated with $P$ and $C$ are not updated
during the $t$-th iteration.
\end{itemize}
\end{itemize}
For $\bigstar - 1$, observe that if $v$ becomes a $+$ vertex during the $t$-th iteration, 
the iteration terminates with one of the following scenarios:
\begin{itemize}
\item[I.] The iteration terminates with {\bf Matching}. 
This contradicts to the assumption that $t^\dagger$, $t^\ddagger$-th iterations are in the same stage, i.e.,
no {\bf Matching} occurs during the $t$-th iteration.
\item[II.] The iteration terminates with {\bf Cycle}. 
The vertex $v$ is connected to the cycle newly added to $\mathcal O$ 
via an even-sized alternating path 
consisting of tight edges in tree $T$ at the start of the next (i.e., $(t+1)$-th) iteration.
\item[III.] The iteration terminates with {\bf Claw}. 
The vertex $v$ becomes a $+$ vertex of  
tree $T$ of the next (i.e., $(t+1)$-th) iteration.
This is due to the following reasons.
After {\bf Claw}, the 
algorithm expands the center vertex of 
newly made claw $W$ by {\bf Expansion} 
in the next iteration. 
Then, there exists an even-sized 
alternating path $P_{W}$ from
$r$ to $v$ consisted of tight edges in 
the newly constructed tree $T$.
Furthermore, edges in $P_{W}$ are continuously added to $T$
by {\bf Grow} without modifying parameter $y$ in Step 2
until $v$ becomes a $+$ vertex in $T$.
This is because {\bf Claw} and {\bf Cycle} are
impossible to occur due to 
Claim \ref{claim1}.
\end{itemize}
For $\bigstar - 2$, in order to derive a contradiction,
assume that a vertex $v$ violates $\bigstar - 2$ at some iteration, 
i.e. the algorithm parameters
associated to the even-sized alternating path $P$
and the cycle $C$ in the statement of $\bigstar-2$
are updated 
during the iteration.
Observe that the algorithm parameters are 
updated due to one of the following scenarios:
\begin{itemize}
\item[I.] The cycle $C$ is contracted. 
If $v$ is in $C$, $v$ no longer remains in $V^\dagger$ and 
contradicts to the assumption that
$v$ remains in $V^\dagger$.
If $v$ is not in $C$, $v$ becomes a $+$ vertex in tree $T$ 
after continuously adding edges of $P$ by {\bf Grow} 
without modifying parameter $y$ due to Claim \ref{claim1}.
This contradicts to the assumption of $\bigstar - 2$ that
$v$ is not in tree $T$ during the $t$-th iteration.\item[II.] A vertex in $C$ is added to tree $T$.
Then, {\bf Matching} occurs,
i.e. the new stage starts.
This contradicts to the assumption that 
$t^\dagger$, $t^\ddagger$-th iterations are in the same stage.
\item[III.] An edge in $P$ is added to tree $T$.
Then, there exists a vertex $u$ in $P$ that first became a $-$ vertex among vertices in $P$,
and it either 
(a) has an even-sized alternating path $P^{\prime}$ to $C$ consisting of tight edges or
(b) has an odd-sized alternating path $P^{\prime}$ to $v$ consisting of tight edges.
For (a),
the edges in $P^{\prime}$ are 
continuously added to $T$ 
without modifying parameter $y$ by Claim \ref{claim1} 
and {\bf Matching} occurs.
This contradicts to the assumption again. 
For (b), $P^{\prime}$ 
are added to $T$
without modifying parameter $y$ due to Claim \ref{claim1},
and $v$ is added to tree $T$ as a $+$ vertex.
This contradicts to the assumption of $\bigstar - 2$ that
$v$ is not in tree $T$ during the $t$-th iteration.
\end{itemize}
Therefore, $\bigstar$ holds.
One can observe that there exists $ t^\ast \in (t^\dagger,t^\ddagger) $ 
such that at the $t^\ast$-th iteration, $v$ last becomes a $+$ vertex before the $t^\ddagger$-th iteration, 
i.e. $v$ is not in tree $T$ during $t$-th iteration for $t^\ast<t<t^\ddagger$.
Then $v$ is connected to some cycle 
$C$ in $\mathcal O$ via an even length alternating path $P$ at $(t^\ast + 1)$-th iteration and
such path and cycle remains unchanged during $t$-th iteration for $t^\ast<t\leq t^\ddagger$ 
due to $\bigstar$.
This completes the proof of Claim \ref{claim2}.

\end{document}